\documentclass[a4paper]{article}
\usepackage{amsmath,amsfonts,amsthm,amssymb,verbatim,graphicx,color}
\usepackage{nicefrac, xfrac}
\usepackage{enumerate}

\newtheorem{lemma}{Lemma}
\newtheorem{theorem}{Theorem}
\newtheorem{proposition}{Proposition}

\newtheorem{observation}{Observation}
\newtheorem{case}{Case}

\title{Token sliding independent set reconfiguration on graphs with few $P_4$'s}

\author{Luc\'{\i}a Busolini \footnote{Departamento de Matemática, Facultad de Ciencias Exactas y Naturales, Universidad de Buenos Aires (UBA) and Instituto de Cálculo, UBA-CONICET, Buenos Aires, Argentina. e-mail: lucia.busolini@gmail.com} \and Mario Valencia-Pabon\footnote{Universit\'e de Lorraine, LORIA, Nancy, France. e-mail: mario.valencia@loria.fr}}

\date{}

\begin{document}
\maketitle

\begin{abstract}
We consider the INDEPENDENT SET RECONFIGURATION problem under the Token Sliding rule. Let $I$ be an independent set of a simple undirected graph $G$. Suppose that each vertex of $I$ has a token placed on it. The tokens are allowed to be moved, one at a time, by sliding along the edges of $G$, so that after each move, the vertices having tokens always form an independent set of $G$. The problem we deal is to decide if we can transform $I$ into $I'$ through a sequence of steps, each of which involves substituting a vertex in the current independent set with one of its neighbours to obtain another independent set. This problem of determining if one independent set of a graph “is reachable” from another independent set of it is known to be PSPACE-hard even for split graphs, planar graphs, and graphs of bounded treewidth. Polynomial time algorithms have been obtained for certain graph classes like trees, interval graphs, claw-free graphs, bipartite permutation graphs, block graphs, and cographs. We present a polynomial time algorithm for the problem on $P_4$-tidy graphs and $(q,q-4)$-graphs, both families of graphs generalizing cographs. 

\end{abstract} 

\section{Introduction}
A reconfiguration problem is a problem of the following type: we are given an instance of a decision problem, two feasible solutions $S$, $T$, and a local modification rule. The question is whether $S$ can be transformed to $T$ by repeated applications of the modification rule in a way that maintains the solution feasible at all times. Due to their numerous applications, reconfiguration problems have attracted much interest in the literature, and reconfiguration versions of standard problems like satisfiability (see the surveys \cite{Heu13,Ni18} and references therein), the independent set problem \cite{HD05,WR18,Bonsma16,IDHP11}, dominating set problem \cite{FHOU15,BJ21}, vertex cover problem \cite{INZ16,MNR18,IDHP11}, matching problem \cite{BBH19,IDHP11,SS21}, and vertex colouring problem \cite{HD05,CHJ08}, have been widely studied.

Among reconfiguration problems on graphs, Independent Set Reconfiguration is certainly the most well-studied. The complexity of this problem depends heavily on the rule specifying the allowed reconfiguration moves. The main reconfiguration rules that have been studied for Independent Set Reconfiguration are Token Addition \& Removal (TAR) \cite{KMM12, MNR17}, Token Jumping (TJ) \cite{BKW14, Bonsma16, BMP17, IDHP11, IKO14, IKO14-2}, and Token Sliding (TS) \cite{BoBo17, DDF15, FHOU15, HD05, LM18, BKL21, FP25}. In all rules, we are required to keep the current set independent at all times. TAR allows us to add or remove any vertex in the current set, as long as the set’s size is always higher than a predetermined threshold. TJ allows to exchange any vertex in the set with any vertex outside it (thus keeping the size of the set constant at all times). Finally, under TS, we are allowed to exchange a vertex in the current independent set with one of its neighbors, that is, we are allowed to perform a TS move only if the two involved vertices are adjacent.

The Independent Set Reconfiguration problem has been intensively studied under all three rules. Because the problem is PSPACE-complete in general for all three rules \cite{KMM12}, this has motivated the study of its complexity in restricted classes of graphs, with an emphasis on graphs where Independent Set is polynomial-time solvable, such as chordal graphs, bipartite graphs and cographs.

In this paper, we focus in the Independent Set Reconfiguration problem under the TS rule which we call {\em TS-ISR}. The problem TS-ISR was first observed to be PSPACE-complete for general graphs by Hearn and Demaine \cite{HD05}. In fact, their result implies that the problem is PSPACE-complete even for subcubic planar graphs (see \cite{KMM12}). Later, the problem was shown to be PSPACE-complete for perfect graphs by Kamiński, Medvedev and Milanič \cite{KMM12}, and this was further improved by Lokshtanov and Mouawad \cite{LM18}, who showed that the problem remains PSPACE-hard even for bipartite graphs. It was shown by Belmonte et al. \cite{BKL21} that the problem is PSPACE-hard even in the case of split graphs. Note that split graphs form a subclass of chordal graphs and even hole-free graphs, and hence the problem is PSPACE-complete for chordal graphs and even-hole free graphs as well. Wrochna \cite{WR18} showed that the problem is PSPACE-complete when restricted to graphs of bounded bandwidth, which implies that the problem is PSPACE-complete for graphs of bounded treewidth, or in fact bounded pathwidth.

Demaine et al. \cite{DDF15} showed that TS-ISR is polynomial time solvable for trees and such a result has been generalized recently to block graphs by Francis and Prabhakaran \cite{FP25}. Polynomial time algorithms for the problem were obtained for claw-free graphs by Bonsma, Kamiński and Wrochna \cite{BKW14}, for cographs by Kamiński, Medvedev and Milanič \cite{KMM12}, for interval graphs by Bonamy and Bousquet \cite{BoBo17}, and for bipartite permutation graphs and bipartite distance-hereditary graphs by Fox-Epstein et al. \cite{FHOU15}.

The main result of this paper is to to show that TS-ISR is polynomial time solvable for $P_4$-tidy graphs and for $(q,q-4)$-graphs, two subclasses of perfect graphs that generalizes cographs.


\section{Definitions and preliminary results}
We use standard graph-theoretic terminology. Let $G = (V,E)$ be a graph. We will denote by $V(G)$ the vertex set $V$, by $E(G)$ the edge set $E$, and by $\overline{G}$ the complement graph of $G$. Given a subset of vertices $X \subset V$, we will denote by $G[X]$ the subgraph of $G$ induced by $X$. The complete graph on $n$ vertices will be denoted by $K_n$ and the stable set of $n$ vertices by $S_n$. The {\em neighborhood} of a vertex $v$ is denoted by $N_G(v)$, that is the set of vertices that are adjacent to $v$, and $N_G[v]$ denotes the \emph{closed neighborhood} of $v$, which is $N_G(v)\cup \{v\}$. Two vertices will be said to be {\em true twins} if they are adjacent and have the same neighborhood, and {\em false twins} if they are non-adjacent but have the same neighbors. 

An induced path on $k$ vertices is denoted by $P_k$. Vertices of degree one (resp. two) in $P_k$ will be called endpoints (resp. midpoints). An induced subgraph of $G$ isomorphic to $P_k$ is simply said to be a $P_k$ in $G$. A chordless cycle on $k$ vertices is denoted by $C_k$.

A {\em cograph} is a graph that does not contain $P_4$ as an induced subgraph \cite{CLS81}. Several generalizations of cographs have been defined in the literature, such as {\em $P_4$-sparse} \cite{Hoang85}, {\em $P_4$-lite} \cite{JO89-1}, {\em $P_4$-extendible} \cite{JO91} and {\em $P_4$-reducible} graphs \cite{JO89-2}. A graph class generalizing all of them is the class of {\em $P_4$-tidy} graphs \cite{GRT97}. Let $G$ be a graph and $A$ a $P_4$ in $G$. A {\em partner} of $A$ is a vertex $v$ in $G \setminus A$ such that $A \cup \{v\}$ induces at least two $P_4$s in $G$. A graph $G$ is $P_4$-sparse if no induced $P_4$ has a partner and $P_4$-tidy if every induced $P_4$ has at most one partner. Another generalization of $P_4$-sparse graphs are $(q, q-4)$-graphs. A graph is a $(q, q-4)$-graph if no set of at most $q$ vertices induces more than $q - 4$ distinct $P_4$’s \cite{BO98}. In this sense, the cographs are precisely the $(4,0)$-graphs and $P_4$-sparse graphs coincide with the $(5,1)$-graphs. However, there is no containment relationship between the classes $P_4$-tidy and $(q, q-4)$-graphs.

Let $G_1 = (V_1, E_1)$ and $G_2 = (V_2, E_2)$ be two graphs such that $V_1 \cap V_2 = \emptyset$. The {\em union} of $G_1$ and $G_2$ is the graph $G_1 \cup G_2 = (V_1 \cup V_2, E_1 \cup E_2)$. The {\em join} of $G_1$ and $G_2$ is the graph $G_1 \vee G_2 = (V_1 \cup V_2, E_1 \cup E_2 \cup V_1 \times V_2)$. That is, the vertex set of $G_1 \vee G_2$ is $V_1 \cup V_2$ and its edge set is $E_1 \cup E_2$ plus all the possible edges with an endpoint in $V_1$ and the other one in $V_2$. 

Cographs can be built from isolated vertices by using these two operations.

\begin{theorem}(\cite{CLS81}) 
Every non-trivial cograph is either a union or join of two smaller cographs.
\end{theorem}
    
$P_4$-tidy graphs have also a useful decomposition theorem. First, we must to introduce some concepts as follows.

Let $G = (V , E)$ be a graph. Let $F = \{e \in E : e \text{ belongs to an induced $P_4$ of $G$}\}$. Let $G_p =(V ,F)$. A connected component of $G_p$ having exactly one vertex is called a {\em weak vertex}. Any connected component of $G_p$ distinct from a weak vertex is called a {\em $p$-component} of $G$. A graph $G$ is {\em $p$-connected} if it has only one $p$-component and no weak vertices \cite{BO99}.

A $p$-connected graph $G = (V , E)$ is {\em $p$-separable} if $V$ can be partitioned into two sets $(C , S)$ such that each $P_4$ that contains vertices from $C$ and from $S$ has its midpoints in $C$ and its endpoints in $S$. We will call it a {\em $p$-partition}. If such a partition exists, then it is unique \cite{JO95}.

An {\em urchin} (resp. {\em starfish}) of size $k$, $k \geq 2$, is a $p$-separable graph with $p$-partition $(C , S)$, where $C = \{c_1,\ldots, c_k\}$ is a clique; $S = \{s_1,\ldots, s_k\}$ is a stable set; $s_i$ is adjacent to $c_i$ if and only if $i = j$ (resp. $i \neq j$).

A {\em quasi-urchin} (resp. {\em quasi-starfish}) of size $k$ is a graph obtained from an urchin (resp. starfish) of size $k$ by replacing at most one vertex by $K_2$ or $S_2$. Note that the new vertices result in true or false twins, respectively, and they are in the same set of the new $p$-partition $(C^*, S^*)$. The elements of $S^*$ are called the {\em legs} and $C^*$ is called the {\em body} of the quasi-starfish or
quasi-urchin.

Note that there are five possible quasi-starfishes of size two, and they are also the five possible quasi-urchins of size two: $P_4$, $P$, $\overline{P}$, fork and kite (see Fig. \ref{flavia-fig}). To avoid ambiguity, we will consider these five graphs as quasi-starfishes, while quasi-urchins will be always of size at least three.

\begin{figure}[htp]
 \centering
 \includegraphics[scale=.5]{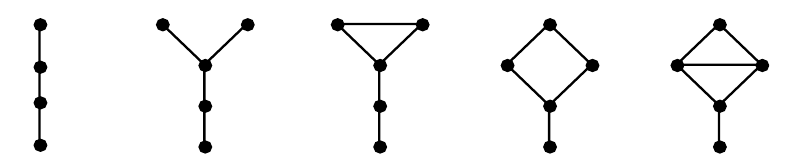}
  \caption{Possible quasi-starfishes of size two. From left to right: $P_4$, fork, $\overline{P}$, $P$ and kite.}
 \label{flavia-fig}
\end{figure} 

When considering quasi-urchins and quasi-starfishes, we have ten kinds of them. We will call {\em type $1$} (resp. {\em type $2$}) the urchins (resp. starfishes); {\em type $3$} (resp. {\em type $4$}) the urchins (resp. starfishes), where a vertex in the body was replaced by $K_2$; {\em type $5$} (resp. {\em type $6$}) the urchins (resp. starfishes), where a vertex in the body was replaced by $S_2$; {\em type $7$} (resp. {\em type $8$}) the urchins (resp. starfishes), where a leg was replaced by $K_2$; and {\em type $9$} (resp. {\em type $10$}) the urchins (resp. starfishes), where a leg was replaced by $S_2$. Recall that graphs of odd type have always size at least three and, with this condition, the ten types form a partition over the family of quasi-urchins and quasi-starfishes.

Let $G_1 = (V_1, E_1)$ and $G_2 = (V_2, E_2)$ be two graphs with $V_1 \cap V_2 = \emptyset$, such that $G_1$ is $p$-separable with partition $(V^1_1 , V^2_1)$. Consider the graph with vertex set $V_1 \cup V_2$ and edge set $E_1 \cup E_2 \cup \{xy : x \in V^1_1 , y \in V_2\}$. We shall denote this graph by $G_1 \veebar G_2$. 

\begin{theorem}(\cite{JO95}). 
\label{th-p}
Every graph $G$ either is $p$-connected or can be obtained uniquely from its $p$-components and weak vertices by a finite sequence of $\cup$, $\vee$ and $\veebar$ operations. Moreover, such a decomposition can be computed in polynomial time on the size of $G$.
\end{theorem}

\begin{proposition}(\cite{GRT97}).
\label{pro$p$-tidy} 
A graph $G$ is $P_4$-tidy if and only if every $p$-component is isomorphic to either $P_5$ or $\overline{P_5}$ or $C_5$ or a quasi-starfish or a quasi-urchin. Quasi-starfishes and quasi-urchins are the $p$-separable $p$-components of $G$.
\end{proposition}

The following result characterize the $p$-connected $(q,q-4)$-graphs. 
\begin{theorem}(\cite{BO98}).
\label{th-q}
Let $G = (V, E)$ be $p$-connected.
\begin{itemize}
\item[(a)] If $G$ is a $(5,1)$-graph then $G$ is an urchin or a starfish.
\item[(b)] If $G$ is a $(7,3)$-graph then $|V| < 7$ or $G$ is an urchin or a starfish.
\item[(c)] If G is a $(q,q-4)$-graph, $q = 6$ or $q \geq 8$, then $|V| < q$.
\end{itemize}
\end{theorem}

Let $G$ be a graph and $A, B$ independent sets of $G$.
If there exists an edge $vw$ in $G$ such that $A \setminus B = \{v\}$ and $B \setminus A = \{w\}$, we say that $B$ can be obtained from $A$ by \emph{sliding} the token on $v \in A$ to $w$ along the edge $vw$.
If there is a sequence $I_1, I_2, \dots, I_l$ of independent sets of $G$ such that $A=I_1$, $B=I_l$ and, for every $i \in \{1,2, \dots, l-1\}$, $I_{i+1}$ can be obtained from $I_i$ by sliding a token, we say that $B$ is \emph{reachable} from $A$ and we denote it by $A \overset{G}{\leftrightsquigarrow} B$.

\begin{observation}\label{obs:SameSize}
    Let $G$ be a graph and $A, B$ independent sets of $G$. If $A \overset{G}{\leftrightsquigarrow} B$, $|A|=|B|$.
\end{observation}

\begin{observation}\label{obs:connectionAndSize1}
    Let $G$ be a graph and $A, B$ independent sets of $G$ with $|A|=|B|=1$. Then, $A \overset{G}{\leftrightsquigarrow} B$ if and only if $a\in A$ and $b \in B$ are in the same connected component of $G$.
    Moreover, if $G$ is connected, $A \overset{G}{\leftrightsquigarrow} B$.
\end{observation}

\section{Token sliding on $P_4$-tidy graphs}

\begin{lemma}
    Let $G$ be $P_5$, $\overline{P_5}$ or $C_5$, and let $A, B$ be independent sets in $G$.
    Then, $A \overset{G}{\leftrightsquigarrow} B$ if and only if $|A|=|B|$.
\end{lemma}
\begin{proof}    
By Observation~\ref{obs:SameSize}, we only have to prove that $|A|=|B|$ implies $A \overset{G}{\leftrightsquigarrow} B$ in each case.
Moreover, since $G$ is connected in the three cases, by Observation~\ref{obs:connectionAndSize1}, if $|A|=|B|=1$, $A \overset{G}{\leftrightsquigarrow} B$. Thus, in the following cases, we can suppose that the independent sets has size at least $2$.

If $G=P_5$, we can numbered the vertices $v_1, v_2, v_3, v_4, v_5$ in such a way that $v_i$ is adjacent to $v_{i+1}$, for each $i \in \{1,2,3,4\}$.
There is only one independent set of size $3$, $\{v_1, v_3, v_5\}$. Thus, if $|A|=3$, $A \overset{G}{\leftrightsquigarrow} B$ if and only if $A=B$ if and only if $|A|=|B|=3$.
Otherwise, if $|A|=2$, $A=\{v_i, v_j\}$ with $i<j$, let us prove that $A \overset{G}{\leftrightsquigarrow} \{v_1, v_5\}$. If $i\neq 1$, we can slide the token in $v_i$ to $v_{i-1}$ and then to $v_{i-2}$, if necessary, and we get that $A \overset{G}{\leftrightsquigarrow} \{v_1, v_j\}$. Analogously, if $j\neq 5$, we can slide the token in $v_j$ to $v_{j+1}$ an then to $v_{j+2}$, if necessary, and we get that $A \overset{G}{\leftrightsquigarrow} \{v_1, v_5\}$.
Then, if $|A|=|B|=2$, both $A \overset{G}{\leftrightsquigarrow} \{v_1, v_5\}$ and $B \overset{G}{\leftrightsquigarrow} \{v_1, v_5\}$. Therefore, if $|A|=2$, $A \overset{G}{\leftrightsquigarrow} B$ if and only if $|A|=|B|=2$.

If $G=\overline{P_5}$, we can numbered the vertices $v_1, v_2, v_3, v_4, v_5$ in such a way that $v_i$ is nonadjacent to $v_{i+1}$, for each $i \in \{1,2,3,4\}$.
There are no independent sets of size $3$, and there are exactly $4$ independent sets of size $2$, $\{v_1, v_2\}$, $\{v_2, v_3\}$, $\{v_3, v_4\}$, $\{v_4, v_5\}$, notice that we can reach $\{v_{i+1}, v_{i+2}\}$ from $\{v_{i}, v_{i+1}\}$ by sliding the token in $v_i$ to $v_{i+2}$, for each $i \in \{1,2,3\}$. Therefore, if $|A|=2$, $A \overset{G}{\leftrightsquigarrow} B$ if and only if $|A|=|B|=2$.

If $G=C_5$, we can numbered the vertices $v_1, v_2, v_3, v_4, v_5$ in such a way that $v_i$ is adjacent to $v_{i+1}$, for each $i \in \{1,2,3,4,5\}$, where subindex are modulus $5$.
There are no independent sets of size $3$, and there are exactly $5$ independent sets of size $2$, $\{v_1, v_3\}$, $\{v_1, v_4\}$, $\{v_2, v_4\}$, $\{v_2, v_5\}$, and $\{v_3, v_5\}$. Notice that, in this order, we can reach one from the previous one by sliding a token. Therefore, if $|A|=2$, $A \overset{G}{\leftrightsquigarrow} B$ if and only if $|A|=|B|=2$.
\end{proof}

\begin{lemma}\label{lemma:TrueTwins}
    Let $G$ be a graph with $v, v'\in V(G)$ true twins and $A, B$ independent sets of $G$.
    Let $G'=G[V(G)\setminus\{v'\}]$ and $A', B'$ be independent sets of $G'$ defined as $A'=A$ if $v'\notin A$ and $A'=(A\setminus \{v'\})\cup \{v\}$ if $v'\in A$, and analogously $B'=B$ if $v'\notin B$ and $B'=(B\setminus \{v'\})\cup \{v\}$ if $v'\in B$.
    Then, $A \overset{G}{\leftrightsquigarrow} B$ if and only if $A' \overset{G'}{\leftrightsquigarrow} B'$.
\end{lemma}
\begin{proof}
    Suppose that $A \overset{G}{\leftrightsquigarrow} B$ and let $A=I_1, I_2, \dots, I_l=B$ such that we can obtain $I_{i+1}$ from $I_i$ by sliding a token along an edge of $G$, for every $i$.
    Define $I'_i=I_i$ if $v'\notin I_i$ and $I'_i=(I_i\setminus \{v'\})\cup \{v\}$ if $v'\in I_i$, for each $i$.
    Notice that $A'=I'_1, I'_2, \dots, I'_l=B$ is a sequence such that $I'_{i+1}=I'_i$ or $I'_{i+1}$ can be obtained from $I'_i$ by sliding a token along an edge of $G$, for every $i$. Thus, just by removing from the sequence the independent sets that are repeated, we get a sequence that certifies that $A' \overset{G'}{\leftrightsquigarrow} B'$.

    Suppose that $A' \overset{G'}{\leftrightsquigarrow} B'$ and let $A'=I_1, I_2, \dots, I_l=B'$ such that we can obtain $I_{i+1}$ from $I_i$ by sliding a token along an edge of $G'$, for every $i$.
    If $A'=(A\setminus \{v'\}))\cup\{v\}$, let $I_0=A$ and if $B'=(B\setminus \{v'\}))\cup\{v\}$, let $I_{l+1}=B$.
    The sequence $(I_0,) I_1, \dots, I_l(, I_{l+1})$, where $I_0$ and $I_{l+1}$ are only if necessary, certifies that $A \overset{G}{\leftrightsquigarrow} B$.
\end{proof}

\begin{observation}\label{obs:FalseTwins}
    Let $G$ be a graph with $v, v'\in V(G)$ false twins and $A$ independent set of $G$ such that $v,v'\in A$. If $B$ independent set of $G$ such that $A \overset{G}{\leftrightsquigarrow} B$, $v,v' \in B$.
    Moreover, $A \overset{G}{\leftrightsquigarrow} B$ if and only if $v,v' \in B$ and $A\setminus\{v,v'\} \overset{G'}{\leftrightsquigarrow} B\setminus \{v,v'\}$, where $G'=G[V(G)\setminus(N_G[v]\cup \{v'\})]$.
\end{observation}
\begin{proof}
    Suppose that $v,v' \in B$ and $A\setminus\{v,v'\} \overset{G'}{\leftrightsquigarrow} B\setminus \{v,v'\}$, that is, there is a sequence $A\setminus\{v,v'\}=I_1, \dots, I_l=B\setminus\{v,v'\}$ of independent sets of $G'$ such that we can obtain $I_{i+1}$ from $I_i$ by sliding a token along an edge of $G'$, for each $i\in 1, \dots, l-1$. The sequence $A=I_1\cup \{v,v'\}, I_2\cup \{v,v'\}, \dots, I_l\cup \{v,v'\}=B$ is such that $I_{i+1}\cup \{v,v'\}$ can be reached from $I_i\cup \{v,v'\}$ by sliding a token along an edge of $G$ and thus, it certifies that $A \overset{G}{\leftrightsquigarrow} B$. 
    
    Suppose that $A \overset{G}{\leftrightsquigarrow} B$ and $\{v, v'\}\nsubseteq B$, let $A=I_1, I_2, \dots I_l=B$ such that $I_{i+1}$ can be obtained from $I_i$ by sliding a token along an edge of $G$, for each $i\in 1, \dots, l-1$. Let $i_0$ be the minimum $i\in \{1, \dots, l\}$ such that $\{v, v'\}\nsubseteq I_{i_0}$. 
    Suppose wlog that $v'\in I_{i_0-1}\setminus I_{i_0}$, this implies that $v \in I_{i_0}$, by definition of $i_0$, and let $w \in I_{i_0}\setminus I_{i_0-1}$, but then $w\in N_G(v')=N_G(v)$ and $v\in I_{i_0}$, a contradiction with the fact that $I_{i_0}$ is an independent set. Therefore, $A \overset{G}{\leftrightsquigarrow} B$ implies $\{v, v'\}\subseteq B$.
    Moreover, $\{v, v'\}\subseteq I_i$ for every $i$ and thus, $A\setminus \{v,v'\}=I_1\setminus \{v,v'\}, I_2\setminus \{v,v'\}, \dots, I_l\setminus \{v,v'\}=B\setminus \{v,v'\}$ is a sequence that certifies that $A\setminus\{v,v'\} \overset{G'}{\leftrightsquigarrow} B\setminus \{v,v'\}$.
\end{proof}

\begin{lemma}\label{lemma:urchin}
    Let $G$ be type 1 (that is, an urchin), and let $A, B$ be independent sets in $G$. Then, $A \overset{G}{\leftrightsquigarrow} B$ if and only if $|A|=|B|$.
\end{lemma}
\begin{proof}
    Let $(C, S)$ be the $p$-partition of $G$.
    Let us prove that $A \overset{G}{\leftrightsquigarrow} I$ where $I=\{s_1, s_2, \dots, s_l\}$ and $l=|A|$, by induction in $|A\setminus I|$.
    Suppose wlog that $A\subseteq S$, otherwise, if $c_j\in A \cap C$ for some $j$, we can slice the token in $c_j$ to $s_j$ to get an independent set $A'$, such that $A'\subseteq S$. 
    If $|A\setminus I|=0$, then $A=I$ and there is nothing to prove.
    If $|A\setminus I|\geq 1$, let $j>l$ be such that $s_j \in A$ or $c_j\in A$ and $i\leq l$ such that $s_i\in I\setminus A$. We can slice the token in $s_j$ to $c_j$ (only if $s_j\in A$, otherwise skip this first step), then to $c_i$ and finally to $s_i$. Then, we get an independent set $A'$ such that $A \overset{G}{\leftrightsquigarrow} A'$ and $|A'\setminus I|<|A\setminus I|$, by inductive hypothesis, $A' \overset{G}{\leftrightsquigarrow} I$ and then, $A \overset{G}{\leftrightsquigarrow} I$.
    Then, if $G$ is an urchin and $|A|=|B|=l$, $A \overset{G}{\leftrightsquigarrow} I$ and $B \overset{G}{\leftrightsquigarrow} I$, thus, $A \overset{G}{\leftrightsquigarrow} B$. 
    Therefore, by Observation~\ref{obs:SameSize}, $A \overset{G}{\leftrightsquigarrow} B$ if and only if $|A|=|B|$.
\end{proof}

\begin{lemma}\label{lemma:quasi-urchin}
    Let $G$ be a quasi-urchin, and let $A, B$ be independent sets in $G$.
    \begin{enumerate}[(a)]
        \item If $G$ is type 5 or 9, $A \overset{G}{\leftrightsquigarrow} B$ if and only if $|A|=|B|$.
        
        \item If $G$ is type 3 and $c, c'$ are the false twins of $G$, $A \overset{G}{\leftrightsquigarrow} B$ if and only if $A=B$ or $|A|=|B|$, $\{c, c'\}\nsubseteq A$ and $\{c, c'\}\nsubseteq B$.
        
        \item If $G$ is type 7 and $s, s'$ are the false twins of $G$, $A \overset{G}{\leftrightsquigarrow} B$ if and only if  $|A|=|B|$ and $\{s, s'\}\subseteq A \Leftrightarrow \{s,s'\}\subseteq B$.
       
    \end{enumerate}
\end{lemma}
\begin{proof}    
If $G$ is type 5 or 9, by Lemmas~\ref{lemma:TrueTwins} and \ref{lemma:urchin}, $A \overset{G}{\leftrightsquigarrow} B$ if and only if $|A|=|B|$.

If $G$ is type 3, 
let $(C^*, S^*)$ be the $p$-partition of $G$ and $c'_1, c_1\in C^*$ the false twins.
If $c_1, c'_1 \in A$, by Observation~\ref{obs:FalseTwins}, $A \overset{G}{\leftrightsquigarrow} B$ if and only if $c_1, c'_1\in B$ and $A\setminus \{c_1, c'_1\} \overset{G'}{\leftrightsquigarrow} B\setminus \{c_1, c'_1\} $, where $G'=G[V(G)\setminus(N_G[c_1]\cup\{c'_1\})]$. Since $G'$ is a stable set $A\setminus \{s_1, s'_1\} \overset{G'}{\leftrightsquigarrow} B\setminus \{c_1, c'_1\} $ if and only if $A\setminus \{c_1, c'_1\}=B\setminus \{c_1, c'_1\}$.
Therefore, if $c_1, c'_1 \in A$, $A \overset{G}{\leftrightsquigarrow} B$ if and only if $A=B$.
If $c_1, c'_1$ are not both in $A$ and  there is some $c_j\in A$ (or $c'_1 \in A$), we can slide the token in $c_j$ to $s_j$ (or $c'_1$ to $s_1$) and get an independent set $A'$ such that $A \overset{G}{\leftrightsquigarrow} A'$ and $A'\subseteq S$. Thus, if $c_1, c'_1$ are not both in $A$, we can assume wlog that $A\subseteq S$ and using the same procedure as in the proof of Lemma~\ref{lemma:urchin}, $A \overset{G}{\leftrightsquigarrow} I$ where $I=\{s_1, \dots, s_l\}$ and $l=|A|$. Therefore, $c_1, c'_1$ are not both in $A$, $A \overset{G}{\leftrightsquigarrow} B$ if and only if $c_1, c'_1$ are not both in $B$ and $|A|=|B|$.
Hence, $A \overset{G}{\leftrightsquigarrow} B$ if and only if $A=B$ or $|A|=|B|$, $\{c_1, c_1'\}\nsubseteq A$ and $\{c_1, c_1'\}\nsubseteq B$, where $c_1,c_1'$ are the false twins of $G$.

If $G$ is type 7, 
let $(C^*, S^*)$ be the $p$-partition of $G$ and $s'_1, s_1\in S^*$ the false twins.
If $s_1, s'_1 \in A$, by Observation~\ref{obs:FalseTwins}, $A \overset{G}{\leftrightsquigarrow} B$ if and only if $s_1, s'_1\in B$ and $A\setminus \{s_1, s'_1\} \overset{G'}{\leftrightsquigarrow} B\setminus \{s_1, s'_1\} $, where $G'=G[V(G)\setminus\{s_1, s'_1, c_1\}]$. Since $G'$ is an urchin of size $k-1$ or a $K_2$, by Lemma~\ref{lemma:urchin},  $A\setminus \{s_1, s'_1\} \overset{G'}{\leftrightsquigarrow} B\setminus \{s_1, s'_1\} $ if and only if $|A\setminus \{s_1, s'_1\}|=|B\setminus \{s_1, s'_1\}|$.
Therefore, if $s_1, s'_1 \in A$, $A \overset{G}{\leftrightsquigarrow} B$ if and only if $s_1, s'_1\in B$ and $|A|=|B|$.
If $s_1, s'_1$ are not both in $A$ and there is some $c_j\in A$, we can slide the token in $c_j$ to $s_j$ and get an independent set $A'$ such that $A \overset{G}{\leftrightsquigarrow} A'$ and $A'\subseteq S^*$. If $s_1, s'_1$ are not both in $A$, $A\subseteq S^*$ and $s'_1\in A$, we can slice the token in $s'_1$ to $c_1$ and then to $s_1$ and get an independent set $A'$ such that $A \overset{G}{\leftrightsquigarrow} A'$ and $A'\subseteq S$. Thus, we can assume that $A\subseteq S$ and using the same procedure of Lemma's~\ref{lemma:urchin} proof, $A \overset{G}{\leftrightsquigarrow} I$ where $I=\{s_1, \dots, s_l\}$ and $l=|A|$. Therefore, $s_1, s'_1$ are not both in $A$, $A \overset{G}{\leftrightsquigarrow} B$ if and only if $s_1, s'_1$ are not both in $B$ and $|A|=|B|$.
Hence, $A \overset{G}{\leftrightsquigarrow} B$ if and only if  $|A|=|B|$ and $\{s_1, s_1'\}\subseteq A \Leftrightarrow \{s_1,s_1'\}\subseteq B$.
\end{proof}

\begin{lemma}\label{lemma:starfish}
Let $G$ be type 2 (that is, a starfish) and let $A, B$ be independent sets in $G$. Then, $A \overset{G}{\leftrightsquigarrow} B$ if and only if $|A|=|B|\leq2$ or $A=B$.
\end{lemma}
\begin{proof}
    Let $(C, S)$ be the $p$-partition of $G$.
    If $|A|=1$, this theorem is a consequence of Observation~\ref{obs:connectionAndSize1}.
    If $|A|=2$, let us prove that $A \overset{G}{\leftrightsquigarrow} \{s_1, c_1\}$. 
    If $A=\{s_i, c_i\}$ for some $i\neq 1$, we can slice the token in $c_i$ to $s_1$ and after that, slice the token in $s_i$ to $c_1$.
    If $A=\{s_i, s_j\}$ for some $i\neq j$, we can slice the token in $s_j$ to $c_i$ and we get the same situation than before.
    In both cases, this procedure proves that $A \overset{G}{\leftrightsquigarrow} \{s_1, c_1\}$. Analogously, $B \overset{G}{\leftrightsquigarrow} \{s_1, c_1\}$ and thus, $A \overset{G}{\leftrightsquigarrow} B$ if $|A|=|B|=2$.
    
    Otherwise, notice that each vertex $c_j\in C$ has degree $2k-2$ and then, every independent set containing $c_j$ has size at most $2$. Thus, if $|A|\geq 3$, $A\subseteq S$ and it is not possible to slide a token, since vertices of $A$ are only adjacent to vertices in $C$ and there is no independent set of size greater to $2$ containing vertices in $C$. 
    Hence, $A \overset{G}{\leftrightsquigarrow} B$ implies $A=B$.
\end{proof}

\begin{lemma} \label{lemma:quasi-starfish}
    Let $G$ be a quasi-starfish and let $A, B$ be independent sets in $G$. 
    \begin{enumerate}[(a)]
        \item If $G$ is type 6 or 10, $A \overset{G}{\leftrightsquigarrow} B$ if and only if 
        \begin{itemize}
            \item $|A|=|B|\leq 2$; or
            \item $|A|=|B|$ and $A\setminus \{v,v'\}=B\setminus \{v,v'\}$, where $v,v'$ are the twins of $G$.
        \end{itemize} 
        \item If $G$ is type 4 and $c, c'$ are the false twins of $G$, $A \overset{G}{\leftrightsquigarrow} B$ if and only if $A=B$ or $\{c, c'\}\nsubseteq A$, $\{c, c'\}\nsubseteq B$, and $|A|=|B|\leq 2$.
        \item If $G$ is type 8 and $s,s'$ be the false twins of $G$, $A \overset{G}{\leftrightsquigarrow} B$ if and only if one of the following holds
        \begin{itemize}
            \item $\{s, s'\}\subseteq A$, $\{s, s'\}\subseteq B$, and $|A|=|B|\leq 3$ or $A=B$;
            \item $\{s, s'\}\nsubseteq A$, $\{s, s'\}\nsubseteq B$, and $|A|=|B|\leq 2$ or $A=B$;
        \end{itemize}
    \end{enumerate}
\end{lemma}
\begin{proof}   
If $G$ is type 6 or 10, by Lemmas~\ref{lemma:TrueTwins} and \ref{lemma:starfish}, $A \overset{G}{\leftrightsquigarrow} B$ if and only if $|A|=|B|\leq 2$ or $A$ and $B$ only differ on a twin, that is, $|A|=|B|$ and $A\setminus \{v,v'\}=B\setminus \{v,v'\}$, where $v,v'$ are true twins of $G$.

If $G$ is type 4, let $(C^*, S^*)$ be the $p$-partition of $G$ and $c_1, c'_1\in C^*$ the false twins.
If $c_1, c'_1 \in A$, by Observation~\ref{obs:FalseTwins}, $A \overset{G}{\leftrightsquigarrow} B$ if and only if $c_1, c'_1\in B$ and $A\setminus \{c_1, c'_1\} \overset{G'}{\leftrightsquigarrow} B\setminus \{c_1, c'_1\} $, where $G'=G[V(G)\setminus(N_G[c_1]\cup\{c'_1\})]$. Since $G'$ has at most one vertex, $A\setminus \{c_1, c'_1\} \overset{G'}{\leftrightsquigarrow} B\setminus \{c_1, c'_1\} $ if and only if $A\setminus \{c_1, c'_1\}=B\setminus \{c_1, c'_1\}$.
Therefore, if $c_1, c'_1 \in A$, $A \overset{G}{\leftrightsquigarrow} B$ if and only if $A=B$.
If $c_1, c'_1$ are not both in $A$. 
As in Lemma~\ref{lemma:starfish}, if $|A|\geq 3$, $A \overset{G}{\leftrightsquigarrow} B$ if and only if $A=B$.
If $|A|=1$, it is a consequence of Observation~\ref{obs:connectionAndSize1}. 
If $|A|=2$, let us prove that $A \overset{G}{\leftrightsquigarrow} \{s_1, c_1\}$.
If $A=\{s_1, c'_1\}$, we can slice the token in $c'_1$ to $s_2$ and after that, to $c_1$.
If $A=\{s_i, c_i\}$ for some $i\neq 1$, we can slice the token in $c_i$ to $s_1$ and after that, slice the token in $s_i$ to $c_1$.
If $A=\{s_i, s_j\}$ for some $i\neq j$, we can slice the token in $s_j$ to $c_i$ and we get the same situation than before.
In all the cases, this procedure proves that $A \overset{G}{\leftrightsquigarrow} \{s_1, c_1\}$. Analogously, $B \overset{G}{\leftrightsquigarrow} \{s_1, c_1\}$ and thus, $A \overset{G}{\leftrightsquigarrow} B$ if $|A|=|B|=2$.

If $G$ is type 8, let $(C^*, S^*)$ be the $p$-partition of $G$ and $s_1, s'_1\in S^*$ the false twins.
If $s_1, s'_1 \in A$, by Observation~\ref{obs:FalseTwins}, $A \overset{G}{\leftrightsquigarrow} B$ if and only if $s_1, s'_1\in B$ and $A\setminus \{s_1, s'_1\} \overset{G'}{\leftrightsquigarrow} B\setminus \{s_1, s'_1\} $, where $G'=G[V(G)\setminus(N_G[s_1]\cup \{s'_1\})]$. Since $G'$ is isomorphic to $K_{1, k-1}$,  $A\setminus \{s_1, s'_1\} \overset{G'}{\leftrightsquigarrow} B\setminus \{s_1, s'_1\} $ if and only if $|A\setminus \{s_1, s'_1\}|=|B\setminus \{s_1, s'_1\}|\leq 1$ or $A\setminus \{s_1, s'_1\}=B\setminus \{s_1, s'_1\}$.
Therefore, if $s_1, s'_1 \in A$, $A \overset{G}{\leftrightsquigarrow} B$ if and only if $s_1, s'_1\in B$ and $|A|=|B|\leq 3$ or $A=B$.
Suppose now that $s_1, s'_1$ are not both in $A$.
As in Lemma~\ref{lemma:starfish}, if $|A|\geq 3$, $A \overset{G}{\leftrightsquigarrow} B$ if and only if $A=B$.
If $|A|=1$, it is a consequence of Observation~\ref{obs:connectionAndSize1}. 
If $|A|=2$, let us prove that $A \overset{G}{\leftrightsquigarrow} \{s_1, c_1\}$.
If $A=\{s'_1, c_1\}$, we can slice the token in $c_1$ to $s_2$, after that, slice the token in $s'_1$ to $c_2$ and then to $s_1$ and finally, slice the token in $s_2$ to $c_1$.
If $A=\{s_i, c_i\}$ for some $i\neq 1$ or $A=\{s_i, s_j\}$ for some $i,j$, the procedure is exactly like the one in the case $G$ is type 4.
In all the cases, this procedure proves that $A \overset{G}{\leftrightsquigarrow} \{s_1, c_1\}$. Analogously, $B \overset{G}{\leftrightsquigarrow} \{s_1, c_1\}$ and thus, $A \overset{G}{\leftrightsquigarrow} B$ if $|A|=|B|=2$.
\end{proof}

Since every graph can be uniquely obtained from its $p$-components and weak vertices by a finite sequence of the operations $\cup$, $\vee$, and $\veebar$ (see Theorem~\ref{th-p}), we need to analyze the relationship between token sliding and these operations.
The $\cup$ and $\vee$ operations were studied in~\cite{KMM12} in order to solve the token sliding problem in cographs. Let $A, B$ be independent sets of a graph $G$.
If $G=G_1\cup G_2$, $A \overset{G}{\leftrightsquigarrow} B$ if and only if $A\cap V_1 \overset{G_1}{\leftrightsquigarrow} B\cap V_1$  and $A \cap V_2\overset{G_2}{\leftrightsquigarrow} B\cap V_2$.
If $G=G_1\vee G_2$, $A \overset{G}{\leftrightsquigarrow} B$ if and only if $|A|=|B|=1$ or $A, B\subseteq V(G_i)$ and $A \overset{G_i}{\leftrightsquigarrow} B$, for some $i \in \{1,2\}$.

It remains to study the case $G=G_1\veebar G_2$ with $G_1$ being a $p$-separable $p$-component of $G$. Since we are interested in solving the problem in $P_4$-tidy graphs, by Proposition~\ref{pro$p$-tidy}, we can assume that $G_1$ is a quasi-urchin or a quasi-starfish.

\begin{proposition}
    Let $G$ be a graph such that $G=G_1\veebar G_2$ where $G_1$ is quasi-urchin or a quasi-starfish with $p$-partition $(V_1^1, V_1^2)$, $A, B$ independent sets of $G$ and let $v,v'$ be the pair of twins in $G_1$, when there is one.
    Then, $A \overset{G}{\leftrightsquigarrow} B$ if and only if one of the following conditions holds:
    \begin{enumerate}[(a)]
        \item \label{item:2verticesInV2} $|A\cap V_2|=|B\cap V_2|\geq 2$, $A\cap V_2 \overset{G_2}{\leftrightsquigarrow} B\cap V_2$ and 
        \begin{itemize}
            \item $A\cap V_1^2 =  B\cap V_1^2$, or
            \item if $v, v' \in V_1^2$ are true twins, $A\setminus \{v,v'\}\cap V_1^2 =  B\setminus\{v,v'\}\cap V_1^2$ and $|A\cap \{v,v'\}|=|B\cap \{v,v'\}|$.
        \end{itemize}        
    
    \noindent In all the following cases, we suppose that $|A\cap V_2|\leq 1$ and $|B\cap V_2|\leq 1$.
    
        \item $G_1$ is type 1, 5, or 9 and 
        \begin{itemize}
            \item $|A|=|B|\leq k$; or 
            \item $|A|=|B|= k+1$ and the vertices $a\in A\cap V_2$ and  $b\in B\cap V_2$ are in the same connected component of $G_2$.
        \end{itemize}
        \item $G_1$ is type 3 and
        \begin{itemize}
            \item $A=B$;
            \item $\{v,v'\}\nsubseteq A$, $\{v,v'\}\nsubseteq B$ and $|A|=|B|\leq k$; or 
            \item  $|A|=|B|= k+1$ and the vertices $a\in A\cap V_2$ and  $b\in B\cap V_2$ are in the same connected component of $G_2$. 
            
            \underline{Note:} It is not necessary to add the condition $\{v,v'\}\nsubseteq A$, $\{v,v'\}\nsubseteq B$ because $A\cap V_2\neq \emptyset$ and $B\cap V_2\neq \emptyset$ implies $v, v' \notin A$ and $v, v' \notin B$
        \end{itemize}
        \item $G_1$ is type 7 and 
        \begin{itemize}
            \item $\{v,v'\}\subseteq A$, $\{v,v'\}\subseteq B$ and $|A|=|B|\leq k-1$; 
            \item $\{v,v'\}\subseteq A$, $\{v,v'\}\subseteq B$, $|A|=|B|= k$ and the vertices $a\in A\cap V_2$ and  $b\in B\cap V_2$ are in the same connected component of $G_2$;
            \item $\{v,v'\}\nsubseteq A$, $\{v,v'\}\nsubseteq B$ and $|A|=|B|\leq k$; or 
            \item $\{v,v'\}\nsubseteq A$, $\{v,v'\}\nsubseteq B$, $|A|=|B|= k+1$ and the vertices $a\in A\cap V_2$ and  $b\in B\cap V_2$ are in the same connected component of $G_2$.
        \end{itemize}
        \item $G_1$ is type 2 and $|A|=|B|\leq 2$ or $A=B$.
        \item $G_1$ is type 6 or 10, and \begin{itemize}
            \item $|A|=|B|\leq 2$; or
            \item $|A|=|B|$ and $A\setminus \{v,v'\}=B\setminus \{v,v'\}$, where $v,v'$ are the twins of $G$.
        \end{itemize} 
        \item $G_1$ is type 4 and 
        \begin{itemize}
            \item $A=B$; or
            \item $\{v,v'\}\nsubseteq A$, $\{v,v'\}\nsubseteq B$ and $|A|=|B|\leq 2$.
        \end{itemize}
        \item $G_1$ is type 8 and 
        \begin{itemize}
            \item $A\cap V_1=B\cap V_1$ and, if there are vertices $a\in A\cap V_2$ and $b\in B\cap V_2$, they are in the same connected component of $G_2$;
            \item $\{v,v'\}\subseteq A$, $\{v,v'\}\subseteq B$ and $|A|=|B|\leq 3$; or
            \item $\{v,v'\}\nsubseteq A$, $\{v,v'\}\nsubseteq B$ and $|A|=|B|\leq 2$.
        \end{itemize}
    \end{enumerate}
\end{proposition}
\begin{proof}
Suppose that $|A\cap V_2|\geq 2$, notice that this implies $A\cap V_1^1=\emptyset$. 
Let $I, I'$ be independent sets of $G$ such that $|I\cap V_2| \geq 2$ and $I'$ is obtained from $I$ by sliding a token from $u$ to $w$. If $u\in V_1$ and $w \in V_2$, since $u$ is adjacent to $w$, $u\in V_1^1$, but is impossible since $I\cap V_2\neq \emptyset$ and $I$ induces an independent set in $G$. Analogously, if $u \in V_2$ and $w\in V_1$, since $u$ is adjacent to $w$, $w\in V_1^1$, but this is impossible since $I'\cap V_2\neq \emptyset$ and $I'$ induces an independent set in $G$, $I'\cap V_1^1=\emptyset$. Therefore, $u, w$ are both in $V_1$ or both in $V_2$ and then, if $A \overset{G}{\leftrightsquigarrow} B$, $|A\cap V_2| = |B\cap V_2|$ and $B\cap V_1^1 =\emptyset$.
Moreover, if $V_1^2$ is an stable set, $A\cap V_1^2= B\cap V_1^2$ since no token can be sliced to $V_1^1$ and otherwise, if $v, v' \in V_1^2$ are true twins, $A\setminus \{v,v'\}\cap V_1^2 =  B\setminus\{v,v'\}\cap V_1^2$ and $|A\cap \{v,v'\}|=|B\cap \{v,v'\}|$.
The proof of (\ref{item:2verticesInV2}) is complete, suppose from now on that $|A\cap V_2|\leq 1$ and $|B\cap V_2|\leq 1$.

\begin{case}\label{case:urchin}
$G_1$ is type 1 (an urchin). 
\end{case}

If $|A|=k+1$, $A$ must have the $k$ vertices in $S=V_1^2$ and one vertex in $V_2$ and then, $A \overset{G}{\leftrightsquigarrow} B$ if and only if $|A|=|B|$ and  the vertices $a \in A\cap V_2$ and $b\in B\cap V_2$ are in the same connected component of $G_2$.

If $|A|=l\leq k$, let us prove that $A \overset{G}{\leftrightsquigarrow} I$ where $I=\{s_1, \dots, s_l\}$ and $l=|A|$.
If $A \subseteq V_1$, applying the same idea of the proof of Lemma~\ref{lemma:urchin}, $A \overset{G}{\leftrightsquigarrow} I$. 
If there is a vertex $a \in A\cap V_2$, there is at least a vertex $w \in C$ such that the only vertex in $S$ adjacent to $w$ is not in $A$, then, we can move the token in $A$ to $w$ and apply the same procedure. 
Therefore, if $|A|\leq k$, 
$A \overset{G}{\leftrightsquigarrow} B$ if and only if $|A|=|B|$.

Moreover, notice that the same argument holds even if $G_1=K_2$, with one vertex in $V_1^1$ and the other in $V_1^2$.


\begin{case} $G_1$ is type 5 or 9. 
\end{case} 
Since a true twin in $G_1$ is also a true twin in $G$, this case is equivalent to Case~\ref{case:urchin}, by Lemma~\ref{lemma:TrueTwins}.

\begin{case}
$G_1$ is type 3. 
\end{case}

Note that $v, v' \in V_1^1$ in the unique $p$-separation of $G_1$.

If $v, v'\in A$, by Observation~\ref{obs:FalseTwins}, $A \overset{G}{\leftrightsquigarrow} B$ if and only if $v,v' \in B$ and $A\setminus \{v,v'\} \overset{G'}{\leftrightsquigarrow} B\setminus \{v,v'\}$, where $G'=G[V(G)\setminus (\{v,v'\}\cup N_G(v))]$. Since $V(G')\subseteq V_1^2$ and $V_1^2$ is a stable set, $A\setminus \{v,v'\} \overset{G'}{\leftrightsquigarrow} B\setminus \{v,v'\}$ is equivalent to $A\setminus \{v,v'\} = B\setminus \{v,v'\}$. Then, if $v,v' \in A$, $A \overset{G}{\leftrightsquigarrow} B$ if and only if $A=B$.

If $v,v'$ are not simultaneously in $A$, the proof is analogous to the one in Case~\ref{case:urchin}.
If $|A|=k+1$, $A \overset{G}{\leftrightsquigarrow} B$ if and only if $|A|=|B|$, $A\cap V_1=B\cap V_1$, and the vertices $a \in A\cap V_2$ and $b\in B\cap V_2$ are in the same connected component of $G_2$. 
If $|A|\leq k$, $A \overset{G}{\leftrightsquigarrow} B$ if and only if $|A|=|B|$.

\begin{case}
$G_1$ is type 7. 
\end{case}
Note that $v, v' \in V_1^2$ in the unique $p$-separation of $G_1$.

If $v, v'\in A$, by Observation~\ref{obs:FalseTwins}, $A \overset{G}{\leftrightsquigarrow} B$ if and only if $v,v' \in B$ and $A\setminus \{v,v'\} \overset{G'}{\leftrightsquigarrow} B\setminus \{v,v'\}$, where $G'=G[V(G)\setminus (\{v,v'\}\cup N_G(v))]$. Since $G'$ can be decomposed 
$G'=G'_1\veebar G_2$ where $G'_1=G_1[V_1\setminus (N_{G_1}[v]\cup\{v'\})]$ is an urchin or $K_2$, we can decide as in Case~\ref{case:urchin}.

If $v,v'$ are not simultaneously in $A$, the proof is analogous to the one in Case~\ref{case:urchin}.
If $|A|=k+1$, $A \overset{G}{\leftrightsquigarrow} B$ if and only if $|A|=|B|$, $A\cap V_1=B\cap V_1$, and the vertices $a \in A\cap V_2$ and $b\in B\cap V_2$ are in the same connected component of $G_2$. 
If $|A|\leq k$, $A \overset{G}{\leftrightsquigarrow} B$ if and only if $|A|=|B|$.

\begin{case}\label{case:starfish}
$G_1$ is type 2 (a starfish).
\end{case}
If $|A|\geq 3$ and $A\subseteq V_1$, no token can be moved, as it was shown in Lemma~\ref{lemma:starfish}. Then, $A \overset{G}{\leftrightsquigarrow} B$ if and only if $A=B$.

If $|A|\geq 3$ and there is a vertex $a\in A\cap V_2$, the others must be in $V_1^2$. Then, $a$ can move along the connected component of $G_2$ it belongs to, but no token can move to $C=V_1^1$, since $A\cap V_1^2$ dominates $V_1^1$ and no token in $A \cap V_1^1$ can move also, since $V_1^1$ is a stable set and $A\cap V_2$ dominates $V_1^1$. Then, $A \overset{G}{\leftrightsquigarrow} B$ if and only if $|A|=|B|$, $A\cap V_1=B\cap V_1$, and the vertices $a \in A\cap V_2$ and $b\in B\cap V_2$ are in the same connected component of $G_2$.

If $|A|=2$, let us prove that $A \overset{G}{\leftrightsquigarrow} I$ where $I=\{s, c\}$ is an independent set in $G_1$ with $s \in S$ and $c\in C$.
If $A \subseteq V_1$, applying the same procedure as in Lemma~\ref{lemma:quasi-starfish}'s proof, $A \overset{G}{\leftrightsquigarrow} I$. 
If there is a vertex $a\in A \cap V_2$ and $a' \in A\cap V_1^2$, let $c' \in C$ be nonadjacent to $a'$, then we can move the token in $a'$ to $c'$ and apply the previous case. Therefore, if $|A|=2$, $A \overset{G}{\leftrightsquigarrow} B$ if and only if $|A|=|B|=2$.

If $|A|=1$, $A \overset{G}{\leftrightsquigarrow} B$ if and only if $|A|=|B|=1$, by Observation~\ref{obs:connectionAndSize1}.

\begin{case}
$G_1$ is type 6 or 10. 
\end{case}
Since a true twin in $G_1$ is also a true twin in $G$, this case is equivalent to Case~\ref{case:starfish}, by Lemma~\ref{lemma:TrueTwins}.

\begin{case}
$G_1$ is type 4. 
\end{case}

Note that $v, v' \in V_1^1$ in the unique $p$-separation of $G_1$.

If $v, v'\in A$, by Observation~\ref{obs:FalseTwins}, $A \overset{G}{\leftrightsquigarrow} B$ if and only if $v,v' \in B$ and $A\setminus \{v,v'\} \overset{G'}{\leftrightsquigarrow} B\setminus \{v,v'\}$, where $G'=G[V(G)\setminus (N_G[v]\cup \{v'\})]$. Since $V(G')\subseteq V_1^2$ and $V_1^2$ is a stable set, $A\setminus \{v,v'\} \overset{G'}{\leftrightsquigarrow} B\setminus \{v,v'\}$ is equivalent to $A\setminus \{v,v'\} = B\setminus \{v,v'\}$. Then, if $v,v' \in A$, $A \overset{G}{\leftrightsquigarrow} B$ if and only if $A=B$.

If $v,v'$ are not simultaneously in $A$, the proof is analogous to Case~\ref{case:starfish}.

\begin{case}
$G_1$ is type 8. 
\end{case}
Note that $v, v' \in V_1^2$ in the unique $p$-separation of $G_1$.

If $v, v'\in A$, by Observation~\ref{obs:FalseTwins}, $A \overset{G}{\leftrightsquigarrow} B$ if and only if $v,v' \in B$ and $A\setminus \{v,v'\} \overset{G'}{\leftrightsquigarrow} B\setminus \{v,v'\}$, where $G'=G[V(G)\setminus (\{v,v'\}\cup N_G(v))]$. Since $G'$ is isomorphic to $K_{1, k-1}$ where the universal vertex is also universal in $G$ and the other $k-1$ are pendant vertices, $A\setminus \{v,v'\} \overset{G'}{\leftrightsquigarrow} B\setminus \{v,v'\}$ is equivalent to $|A\setminus \{v,v'\}|=|B\setminus \{v,v'\}|\leq 1$, or $|A\setminus \{v,v'\}| = |B\setminus \{v,v'\}|$, $A\cap V_1=B\cap V_1$ and if there are vertices $a\in A\cap V_2$ and $b\in B\cap V_2$, they are in the same connected component of $G_2$. 
Then, if $v,v' \in A$, $A \overset{G}{\leftrightsquigarrow} B$ if and only if $v, v'\in B$, $|A|=|B|$, and ($|A|=|B|\leq 3$) or ($A\cap V_1=B\cap V_1$ and if there are vertices $a\in A\cap V_2$ and $b\in B\cap V_2$, they are in the same connected component of $G_2$). 

If $v, v'$ are not simultaneously in $A$, the proof is analogous to Case~\ref{case:starfish}.
\end{proof}

By Theorem \ref{th-p} and Proposition \ref{pro$p$-tidy}, and as a consequence of previous results in this section, we obtain the following:

\begin{theorem}
There is a polynomial-time algorithm deciding TS-ISR in $P_4$-tidy graphs.
\end{theorem}

\section{Token sliding on bounded $p$-components graphs and in $(q,q-4)$ graphs}
We start this section with two general results.


\begin{lemma}
\label{l1-p}
Let $G = G_1 \veebar G_2$ such that $G_1$ is $p$-separable with $p$-partition $(V^1_1 , V^2_1)$. Let $A$ and $B$ two independent sets in G. If $|A \cap V(G_2)| \geq 2$ then, $A \overset{G}{\leftrightsquigarrow} B$ if and only if $A \cap V(G_2)\overset{G_2}{\leftrightsquigarrow} B \cap V(G_2)$ and $A \cap V_1^2\overset{G[V_1^2]}{\leftrightsquigarrow} B \cap V_1^2$.
\end{lemma}

\begin{proof}
Notice first that no vertex in $A \cap V(G_2)$ can slide to $G_1$. In fact, let $v_1,v_2$ to vertices in $A \cap V(G_2)$.  As all the vertices in $G_2$ are joined to all vertices in $V_1^1$ and to no vertex in $V_1^2$, the only way for $v_1$ or $v_2$ to slide to a vertex in $G_1$ is via a vertex in $V_1^1$ which is not possible. Moreover, notice that for a similar argument, $A \cap V_1^1 = \emptyset$ and no vertex in $A \cap V_1^2$ can slide to $G_2$ because the only way for one of these vertices to slide to $G_2$ is passing through a vertex in $V_1^1$ which is not possible. 
Finally, notice that $A \overset{G}{\leftrightsquigarrow} B$ implies $|A \cap V(G_1)| = |B \cap V(G_1)|$ and $|A \cap V(G_2)| = |B \cap V(G_2)|$, if $|A\cap V(G_2)|\geq 2$.
\end{proof}

\begin{lemma}
\label{l2-p}
Let $G = G_1 \veebar G_2$ such that $G_1$ is $p$-separable with $p$-partition $(V^1_1 , V^2_1)$. Let $A$ and $B$ two independent sets in G. If $|A \cap V(G_2)| \leq 1$ and $A \overset{G}{\leftrightsquigarrow} B$, $|B \cap V(G_2)| \leq 1$. Moreover, one of the following conditions holds:
\begin{enumerate}[(a)]
    \item \label{item:BothEmpty} If $A \cap V(G_2) = \emptyset$ and $B \cap V(G_2) = \emptyset$ then, $A \overset{G}{\leftrightsquigarrow} B$ if and only if $A \overset{G[V(G_1) \cup \{v\}]}{\leftrightsquigarrow} B$  where $v$ is any vertex of $G_2$. 

    \item \label{item:OneEmpty} If $A \cap V(G_2) \neq \emptyset$ and $B \cap V(G_2) = \emptyset$ then, $A \overset{G}{\leftrightsquigarrow} B$ if and only if $A \overset{G[V(G_1) \cup \{v\}]}{\leftrightsquigarrow} B$  where $v$ is the only vertex in $A\cap V(G_2)$.  
    
    If $A \cap V(G_2) = \emptyset$ and $B \cap V(G_2) \neq \emptyset$ is analogous.
    
    \item \label{item:BothNonEmptyInTheSameComponent} If $A \cap V(G_2) \neq \emptyset$, $B \cap V(G_2) \neq \emptyset$, and there is only one connected component $C$ in $G_2$ such that $A \subseteq V(G_1) \cup V(C)$ and $B \subseteq V(G_1) \cup V(C)$ then, $A \overset{G}{\leftrightsquigarrow} B$ if and only if $(A \setminus \{v\}) \cup \{w\} \overset{G[V(G_1) \cup \{w\}]}{\leftrightsquigarrow} B$, where $v\in A\cap V(G_2)$ and $w\in B\cap V(G_2)$.  

    \item \label{item:BothNonEmptyInDistinctComponents} If $A \cap V(G_2) \neq \emptyset$, $B \cap V(G_2) \neq \emptyset$, and there are two different connected components $C_1$ and $C_2$ of $G_2$ such that $v \in A \cap C_1$ and $w \in B \cap C_2$ then, $A \overset{G}{\leftrightsquigarrow} B$ if and only if $A \overset{G[V(G_1) \cup \{v,w\}]}{\leftrightsquigarrow} B$. 
\end{enumerate}
\end{lemma}

\begin{proof}

By Lemma~\ref{l1-p}, independent sets with more than one token in $V(G_2)$ can only reach other independent sets having the same property. Hence, if $|A \cap V(G_2)| \leq 1$ and $A \overset{G}{\leftrightsquigarrow} B$ then, $|B \cap V(G_2)| \leq 1$.

(\ref{item:BothEmpty}) and (\ref{item:OneEmpty}) Since $G'=G[V(G_1)\cup \{v\}]$ is an induced subgraph of $G$, $A \overset{G'}{\leftrightsquigarrow} B$ implies $A \overset{G}{\leftrightsquigarrow} B$. 
Suppose now that $A \overset{G}{\leftrightsquigarrow} B$, and let $A=I_1, I_2, \dots, I_l=B$ a sequence of independent sets such that $I_{i+1}$ can be obtained from $I_i$ by sliding a token along an edge of $G$. Notice that $|I_i\cap V_2|\leq 1$ for every $i$, by Lemma~\ref{l1-p}. Let $I'_i=I_i$ if $I_i\cap V_2=\emptyset$ and $I'_1 = (I_1\cap V_1)\cup \{v\}$, otherwise. Then, $A=I'_1, I'_2, \dots, I'_l=B$ is a sequence such that $I'_{i+1}=I'_i$ or $I'_{i+1}$ can be obtained from $I'_i$ by sliding a token along an edge of $G'$. Thus, just by removing from the sequence the independent sets that are repeated, we get a sequence that certifies that $A \overset{G'}{\leftrightsquigarrow} B$.

(\ref{item:BothNonEmptyInTheSameComponent}) If $v$ and $w$ are in the same connected component $C$ of $G_2$, the token $v$ in $A$ can be slide to $w$ along a $vw$-path in $G_2$. Therefore, $A \overset{G}{\leftrightsquigarrow} B$ if and only if $(A\setminus \{v\})\cup \{w\} \overset{G}{\leftrightsquigarrow} B$. Using the same procedure as before, we get that $(A\setminus \{v\})\cup \{w\} \overset{G}{\leftrightsquigarrow} B$ if and only if $(A\setminus \{v\})\cup \{w\} \overset{G'}{\leftrightsquigarrow} B$, where $G'=G[V(G_1)\cup \{w\}]$. 


(\ref{item:BothNonEmptyInDistinctComponents})
Since $G'=G[V(G_1)\cup \{v,w\}]$ is an induced subgraph of $G$, $A \overset{G'}{\leftrightsquigarrow} B$ implies $A \overset{G}{\leftrightsquigarrow} B$. 
Suppose now that $A \overset{G}{\leftrightsquigarrow} B$, and let $A=I_1, I_2, \dots, I_l=B$ a sequence of independent sets such that $I_{i+1}$ can be obtained from $I_i$ by sliding a token along an edge of $G$. Since there is no way of moving the token in $v$ to $w$ without moving it to $V_1$, there is $i_0$ such that $I_{i_0}\subseteq V_1$. Let $I'_i=I_i$ if $I_i \subseteq V_1$, $I'_i=(I_i\cap V_1) \cup \{v\}$ if $I_i \nsubseteq V_1$ and $i<i_0$, and $I'_i=(I_i\cap V_1) \cup \{w\}$ if $I_i \nsubseteq V_1$ and $i>i_0$.
Then, $A=I'_1, I'_2, \dots, I'_l=B$ is a sequence such that $I'_{i+1}=I'_i$ or $I'_{i+1}$ can be obtained from $I'_i$ by sliding a token along an edge of $G'$. Thus, just by removing from the sequence the independent sets that are repeated, we get a sequence that certifies that $A \overset{G'}{\leftrightsquigarrow} B$.
\end{proof}

\begin{observation}
\label{ob-p}
Let $G$ be a graph on $n$ vertices and let $A$ and $B$ two independent sets in $G$ with $|A| = |B| = k$. There is a naive $O(kn^{2k})$-time algorithm to decide if $A$ can be reconfigured into $B$ under the Token Slide rule.
\end{observation}

In fact, we can enumerate all $k$-sets of vertices of $G$ in time $O(n^k)$ and test independence in $O(k^2)$ each one of them. After, we can construct a graph $G_k = (V_k,E_k)$ where $V_k$ is the set of independent $k$-sets in $G$ and where two independent $k$-sets $A_k$ and $B_k$ are adjacent if the size of their symmetric difference is exactly equal to two and if $a \in A_k \setminus B_k$ and $b \in B_K \setminus A_k$ are adjacent in $G$. The construction of $G_k$ can be done in $O(kn^{2k})$ time. Finally, test if there is a $AB$-path in $G_k$ can be done in $O(n^{2k})$.  

By Theorem \ref{th-p}, every graph can be obtained from its $p$-components and weak vertices by a finite sequence of operations $\cup$, $\vee$, and $\veebar$ in polynomial time. Moreover, By Lemmas \ref{l1-p} and \ref{l2-p}, and by Observation \ref{ob-p}, we can deduce the following result.
\begin{theorem}
\label{th2-p}
Let $G$ be a graph with bounded $p$-components. Then, there is a polynomial-time algorithm deciding TS-ISR in $G$.
\end{theorem}

By Theorems \ref{th-q} and \ref{th2-p}, and by Lemmas \ref{lemma:urchin} and \ref{lemma:starfish}, we have the following result:

\begin{theorem}
Let $q$ be a fixed parameter with $q\geq 4$ and let $G$ be a $(q,q-4)$-graph. Then, there is a polynomial-time algorithm deciding TS-ISR in $G$.
\end{theorem}

\section{Discussion}
We have proved in this paper that the TS-ISR problem is polynomial time solvable for $P_4$-tidy graphs and for $(q,q-4)$-graphs, two subclasses of perfect graphs that generalizes cographs. We have also shown that the TS-ISR problem can be solved in polynomial time for graphs having bounded $p$-connected components. Bonsma \cite{Bonsma16} proved that under the TJ and TAR rules, the independent set reconfiguration problem on cographs is polynomial time solvable. It will be interesting to extend our results in order to show that under the TJ and TAR rules the independent sets reconfiguration remains polynomial time solvable in graphs with few $P_4$s.

\section{Acknowledgments}
This work was supported by the European Union through the MSCA-SE project QCOMICAL (Grant Agreement ID: 101182520) and by the IRP SINFIN (CNRS-CONICET-UBA, France–Argentine).




\begin{thebibliography}{99}
\bibitem{BO98}
L. Babel, S. Olariu. {\em On the structure of graphs with few $P_4$s}. Discrete Applied Mathematics 84:1–13, 1998.

\bibitem{BO99}
L. Babel, S. Olariu. {\em On the $p$-connectedness of graphs—a survey}. Discrete Applied Mathematics 95(1–3):11–33, 1999.

\bibitem{BKL21}
R. Belmonte, E. J. Kim, M. Lampis, V. Mitsou, Y. Otachi, F. Sikora. {\em Token slidong on split graphs}. Theory Comput Syst 65:662–686, 2021.

\bibitem{BoBo17}
M. Bonamy, N. Bousquet. {\em Token sliding on chordal graphs}. In proc. of WG 2017, volume 10520 of Lecture Notes in Computer Science, pages 127–139, 2017.

\bibitem{BBH19}
M. Bonamy, N. Bousquet, M. Heinrich, T. Ito, Y. Kobayashi, A. Mary, M. Mühlenthaler, K. Wasa. {\em The perfect matching reconfiguration problem}. In proc. of MFCS 2019, volume 138 of LIPIcs, pages 80:1–80:14, 2019.

\bibitem{Bonsma16}
P. S. Bonsma. {\em Independent set reconfiguration in cographs and their generalizations}. Journal
of Graph Theory, 83(2):164–195, 2016.

\bibitem{BKW14}
P. S. Bonsma, M. Kaminski, M. Wrochna. {\em Reconfiguring independent sets in claw-free graphs}. In proc. of SWAT 2014, volume 8503 of Lecture Notes in Computer Science, pages 86–97, 2014.

\bibitem{BJ21}
N. Bousquet, A. Joffard. {\em TS-reconfiguration of dominating sets in circle and circular-arc graphs}. In proc. of FCT 2021, volume 12867 of Lecture Notes in Computer Science, pages 114–134, 2021.

\bibitem{BMP17}
N. Bousquet, A. Mary, A. Parreau. {\em Token jumping in minor-closed classes}. In proc. of FCT 2017, volume 10472 of Lecture Notes in Computer Science, pages 136–149, 2017.

\bibitem{CHJ08}
L. Cereceda, J. van den Heuvel, M. Johnson. {\em Connectedness of the graph of vertex-colourings}. Discrete Mathematics, 308(5-6):913–919, 2008.

\bibitem{CLS81}
D. Corneil, H. Lerchs, L. Stewart Burlingham. {\em Complement reducible graphs}. Discrete Applied Mathematics 3(3):163–174, 1981.

\bibitem{DDF15}
E. D. Demaine, M. L. Demaine, E. Fox-Epstein, D. A. Hoang, T. Ito, H. Ono, Y. Otachi, R. Uehara, T. Yamada. {\em Linear-time algorithm for sliding tokens on trees}. Theor. Comput. Sci., 600:132–142, 2015.

\bibitem{FP25}
M. C. Francis, V. Prabhakaran. {Token sliding independent set reconfiguration on block graphs}. In proc. of FSTTCS 2025, Volume 360 of LIPIcs, pp. 31:1-31:19, 2025.

\bibitem{FHOU15}
E. Fox-Epstein, D. A. Hoang, Y. Otachi, R. Uehara. {\em Sliding token on bipartite permutation graphs}. In proc. of ISAAC 2015, volume 9472 of Lecture Notes in Computer Science, pages 237–247, 2015.

\bibitem{GRT97}
V. Giakoumakis, F. Roussel, H. Thuillier. {\em On $P_4$-tidy graphs}. Discrete Mathematics and Theoretical Computer Science 1:17–41, 1997.

\bibitem{HD05}
R. A. Hearn, E. D. Demaine. {\em PSPACE-completeness of sliding-block puzzles and other problems through the nondeterministic constraint logic model of computation}. Theoretical Computer Science, 343(1-2):72–96, 2005.

\bibitem{Heu13}
J. van den Heuvel. {\em The complexity of change}. Surveys in Combinatorics 2013, volume 409 of London Mathematical Society Lecture Note Series, pages 127–160, 2013.

\bibitem{Hoang85}
C. T. Hoàng. {\em Perfect graphs}. PhD thesis, School of Computer Science, McGill University, 1985.

\bibitem{IDHP11}
T. Ito, E. D. Demaine, N. J.A. Harvey, Ch. H. Papadimitriou, M. Sideri, R. Uehara, Y. Uno. {\em On the complexity of reconfiguration problems}. Theoretical Computer Science, 412(12):1054–1065, 2011.

\bibitem{IKO14}
T. Ito, M. Kaminski, H. Ono, A. Suzuki, R. Uehara, K. Yamanaka. {\em On the parameterized complexity for token jumping on graphs}. In proc. of TAMC 2014, volume 8402 of Lecture Notes in Computer Science, pages 341–351, 2014.

\bibitem{IKO14-2}
T. Ito, M. Kaminski, H. Ono. {\em Fixed-parameter tractability of token jumping on planar graphs}. In proc. of ISAAC 2014, volume 8889 of Lecture Notes in Computer Science, pages 208–219, 2014.

\bibitem{INZ16}
T. Ito, H. Nooka, X. Zhou. {\em Reconfiguration of vertex covers in a graph}. IEICE Transactions on Information and Systems, 99-D(3):598–606, 2016.

\bibitem{JO89-1}
B. Jamison, S. Olariu. {\em A new class of brittle graphs}, Studies in Applied Mathematics 81:89–92, 1989.

\bibitem{JO89-2}
B. Jamison, S. Olariu. {\em P4 -reducible graphs—a class of uniquely tree representable graphs}. Studies in Applied Mathematics 81:79–87, 1989.

\bibitem{JO91}
B. Jamison, S. Olariu. {\em On a unique tree representation for P4 -extendible graphs}. Discrete Applied Mathematics 34:151–164, 1991.

\bibitem{JO95}
B. Jamison, S. Olariu. {\em $p$-components and the homogeneous decomposition of graphs}. SIAM Journal on Discrete Mathematics 8:448–463, 1995.

\bibitem{KMM12}
M. Kamiński, P. Medvedev, M. Milanič. {\em Complexity of independent set reconfigurability problems}. Theor. Comput. Sci., 439:9–15, 2012.

\bibitem{LM18}
D. Lokshtanov, A. E. Mouawad. {\em The complexity of independent set reconfiguration on bipartite graphs}. ACM Transactions on Algorithms (TALG), 15(1):1–19, 2018.

\bibitem{MNR18}
A. E. Mouawad, N. Nishimura, V. Raman, S. Siebertz. {\em Vertex cover reconfiguration and beyond}. Algorithms, 11(2):20, 2018.

\bibitem{MNR17}
A. E. Mouawad, N. Nishimura, V. Raman, N. Simjour, A. Suzuki. {\em On the parameterized complexity of reconfiguration problems}. Algorithmica, 78(1):274–297, 2017.

\bibitem{Ni18}
N. Nishimura. {\em Introduction to reconfiguration}. Algorithms, 11(4):52, 2018. 

\bibitem{SS21}
N. Solomon, Sh. Solomon. {\em A generalized matching reconfiguration problem}. In proc. of ITCS 2021, volume 185 of LIPIcs, pages 57:1–57:20, 2021.

\bibitem{WR18}
M. Wrochna. {\em Reconfiguration in bounded bandwidth and tree-depth}. Journal of Computer and System Sciences, 93:1–10, 2018.

\end{thebibliography}
\end{document}